\documentclass[12pt,draftclsnofoot,journal,letterpaper,onecolumn]{IEEEtran}
\usepackage{graphicx}
\usepackage{multicol}
\usepackage{amssymb}
\usepackage{amsmath}
\usepackage{amsfonts}
\usepackage{amsthm}
\newtheorem{thm}{Theorem}
\newtheorem{thm1}{Theorem}

\newtheorem{cor}[thm1]{Corollary}

\newtheorem*{obv*}{Observation}
\newtheorem{prpn}[thm]{Proposition}
\usepackage{fancyhdr}
\usepackage{lastpage}
\usepackage{tikz}
\usepackage{algorithmicx}
\usepackage{algorithm}
\usepackage{algpseudocode}
\usepackage[numbers,square,sort&compress]{natbib}
\usetikzlibrary{arrows,shapes}
\usetikzlibrary{arrows,automata}
\def\d{\dagger}
\renewcommand{\t}{\widehat}

\newcommand{\mb}{\mathbf}
\newcommand{\hs}[1]{\hspace*{#1}}
\newcommand{\ignore}[2]{\hspace{0in}#2}
\def\sq2{\frac{1}{2}}
\def\trnp{^\mathbf{t}}
\def\<{\left\langle}
\def\>{\right\rangle}
\title{Beamforming for Secure Communication via Untrusted Relay Nodes Using Artificial Noise}
\author{\IEEEauthorblockN{Siddhartha Sarma}
\thanks{S. Sarma is with the Department of Electronic Systems Engineering, Indian Institute of Science, Bangalore, Karnataka - 560012, India (e-mail: \{siddharth\}@dese.iisc.ernet.in).}
}
\begin{document}
\maketitle
\begin{abstract} 
A two-phase beamforming solution for secure communication using untrusted relay nodes is presented. To thwart eavesdropping attempts of relay nodes, we deliberately introduce \ignore{The  introduction of} artificial noise in the source message.
After pointing out the incongruity \ignore{discussing the contradiction involved} in evaluating secrecy rate \ignore{showing the nonconformity of secrecy rate} in our model for certain scenarios, we provide an SNR based frame work for secure communication. \ignore{To prevent eavesdropping} We intend to bring down the SNR at each of the untrusted relay nodes below a certain predefined threshold, whereas, using beamforming we want to boost the SNR at the destination. With this motive optimal scaling vector is evaluated for beamforming phase which not only nullifies the artificial noise transmitted initially, but also maximizes the SNR at the destination. We discuss both the total and individual power constraint scenarios and provide analytical solution for both of them. 
\end{abstract}
\section{Introduction} 
Recently omnipresence of wireless devices prompted the researchers and engineers  to delve into the security issues of wireless communication. As compared to wired medium, ensuring security for wireless medium is more challenging mainly due to the broadcast nature of the transmission. But the ongoing research on \textit{physical layer security} \ignore{has paved the path for has revealed new possibilities for} promises robust and reliable security schemes for wireless communication.
Contrary to conventional cryptographic schemes physical layer security techniques are impregnable as the security is ensured by inherent randomness present in the wireless medium.
\par Physical layer security came to existence through the seminal work of Wyner \cite{wyner}, where he explored the possibility of secure communication without relying on private keys. Later Leung and Hellman \cite{Leung} extended this idea to Gaussian channels. Following their works, 
in last decade several researchers have devised techniques \ignore{come up with schemes} for secure communication  in \ignore{to maximize secrecy rate for} single-hop single and multi-antenna systems. \ignore{On the other hand, multi-hop networks cover a broad spectrum of wireless scenarios.} Recently a considerable amount of research is ongoing to extend those schemes to multi-hop network scenarios. In fact, several works have demonstrated that \textit{cooperative relaying} \cite{dong,zhang10} can \ignore{aid } significantly improve the performance of secure communication. 
\ignore{significance of  pointed out that relays not only improves information capacity but in several cases it can improve secrecy rate also.} One such cooperative scheme is beamforming, where multiple transmitters adjust their gain and phase to improve the signal strength at the destination. Though beamforming was initially developed  for multi-antenna systems, but due to energy and hardware constraints distributed beamforming solution using multiple single antenna node is proposed \cite{dong,zhang10}. In those papers authors have assumed that relay nodes are trustworthy and they beamform towards the destination to defeat (an) external eavesdropper(s).
 But due to the unprotected nature of public ad hoc network (\textit{e.g.} sensor network), an adversary can wiretap those relay nodes to obtain the information transmitted by source. This prompted us to look for a solution of one of the worst adversarial scenarios of secure communication, where one have to maintain the secrecy of data from the participating relay nodes. Secure transmission using untrusted relay nodes has appeared before in \cite{xiang,zhang,jeong12,he2010cooperation}, where secure communication \ignore{of the message(s) was(were)} is obtained by simultaneously transmitting two or more signals. But multiple simultaneous transmission in the above mentioned papers were possible due to mutiple sources and/or multiple jammers. In absence of mutiple sources or jammers source has to rely on artificial noise \cite{goel}. Authors in \cite{yang2013} considered artificial noise based beamforming solution for AF relay network, but unlike them we have untrusted relay nodes and artificial noise is introduced by the source itself.
\par In our work we have considered multiple amplify-and-forward (AF) untrusted relay nodes, who help the source to deliver the message at destination, but at the same time can eavesdrop the ongoing transmission. To prevent eavesdropping the source adds artificial noise in the message and broadcast it in first phase. In second phase source and relay nodes beamform to deliver the signal at the destination. Though \textit{secrecy rate} is the conventional metric to evaluate the performance of such system, but for our model in certain scenarios maximizing secrecy rate results in removal of artificial noise. Therefore, we consider a pragmatic approach based on SNR criteria, where we intend keep the SNR at the untrusted relay nodes below certain predefined threshold. Assuming perfect channel state information (CSI), we evaluate the optimal scaling factor for all the transmitting nodes which not only nullifies the artificial noise but also maximizes the SNR at the destination. We consider both the total and individual power constraints scenario for all transmitting nodes and provide analytical approach for optimal solutions.

\section{System Model \& Problem Formulation}
We consider a two-hop network shown in Figure \ref{fig:model} containing multiple relay nodes ($\mathcal{M}=\{1,2,\ldots,M\}$) aiding the transmission from source node $S$ to destination $D$. Adversary is using the relay nodes to passively eavesdrop the on going transmission. To prevent such eavesdropping source transmits artificial noise along with the actual signal to confuse the adversary. In next phase source along with the relay nodes performs beamforming to deliver the message to the destination. We are interested in maximizing the information rate which can be transmitted from the source to the destination securely.
\par In first phase source transmits a linear combination of two messages $\sqrt{\alpha P_1}x+\sqrt{(1-\alpha)P_1}u$ using power $P_1$, where $\alpha \in [0,1]$, $x$ is the Gaussian distributed message with zero mean and unit variance and $u$ is statistically independent Gaussian distributed artificial noise with identical distribution as $x$. The message received at relay and the destination can be written as:
\[y_i=h_{si}(\sqrt{\alpha P_1}x+\sqrt{(1-\alpha)P_1}u)+z_i,\,i \in \mathcal{M}\bigcup\{D\} \]
where $z_i$ are complex Gaussian noise distributed according to $\mathcal{CN}(0,\sigma^2)$.
In second phase relay nodes along with the source performs beamforming at destination to nullify the artificial noise. 
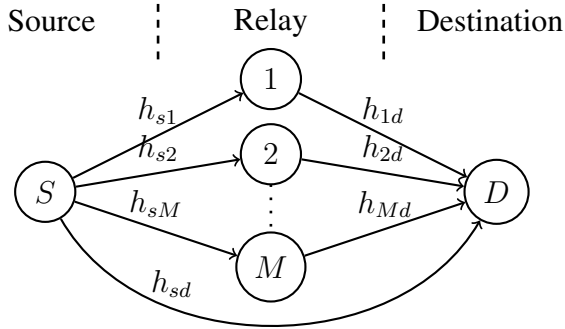
\begin{figure}[!t]
\centering
\begin{tikzpicture}[scale=1]
\tikzstyle{circnode}=[draw,shape=circle,minimum size=0.8cm,style=thick]
\node[circnode] (v0) at (180:3) {$S$};
\node[circnode] (v1) at (90:1.5) {$1$};
\node[circnode] (v2) at (90:0.5) {$2$};
\node[circnode] (v3) at (270:1) {$M$};
\node[circnode] (v4) at (0:3) {$D$};
\node (v5) at (142:3.7) {Source};
\node (v6) at (90:2.25) {Relay};
\node (v7) at (38:3.7) {Destination};
\draw[style=thick,->] (v0) -- (v1)node[pos=0.5,above]{$h_{s1}$};
\draw[style=thick,->] (v0) -- (v2)node[pos=0.5,above]{$h_{s2}$};
\draw[style=thick,->] (v0) -- (v3)node[pos=0.5,above]{$h_{sM}$};
\draw[style=thick,->] (v1) -- (v4)node[pos=0.5,above]{$h_{1d}$};
\draw[style=thick,->] (v2) -- (v4)node[pos=0.5,above]{$h_{2d}$};
\draw[style=thick,->] (v3) -- (v4)node[pos=0.5,above]{$h_{Md}$};
\path (v0) edge[style=thick,bend right=60,->] node[pos=0.3,above]{$h_{sd}$} (v4); 
\draw[loosely dotted,line width=1](0,0.10) -- (0,-.6);
\draw[dashed,line width=1](-1.5,2.5) -- (-1.5,1.75);
\draw[dashed,line width=1](1.5,2.5) -- (1.5,1.75);
\end{tikzpicture}
\caption{Simple two hop network with multiple relays}
\label{fig:model}
\end{figure}
In amplify-and-forward (AF) relaying scheme nodes transmit the received noisy message after scaling it by appropriate factor, $w_i \in \mathbb{C}$. Source simultaneously transmits $\sqrt{\alpha P_1}w_0x-\sqrt{(1-\alpha )P_1}\sum\limits_{i=1}^{M}\frac{w_ih_{si}h_{id}}{h_{sd}}u$. Therefore, the received message at destination can be written as:
\[y_d=\sqrt{\alpha P_1}(w_0h_{sd}+\sum\limits_{i=1}^{M}w_ih_{si}h_{id})x+\sum\limits_{i=1}^{M}h_{id}w_iz_i+z_d
\]
Assuming Maximum Ratio Combining (MRC) \cite{dong}, capacity of the source to destination channel can be written as:
\[C_d\hspace*{-0.1cm}=\hspace*{-0.1cm}\frac{1}{2}\log\left( 1+\hspace*{-0.1cm}\frac{|h_{sd}|^2\alpha P_1}{\sigma^2+\hspace*{-0.15cm}|h_{sd}|^2(1-\alpha )P_1}+\frac{|\mathbf{h}^\d\mathbf{w}|^2\alpha P_1}{\sigma^2(1+\mathbf{w^\d D_hw})}\right)\]
where $\mb{w}=[w_0,w_1,\cdots,w_M]^\d$, $\mb{h}=[h_{sd}, h_{s1}h_{1d},\cdots,h_{sM}$\\$h_{Md}]^\d$,  and $\mathbf{D_h}$ is $(M+1)\times(M+1)$ matrix with diagonal elements $[0,|h_{1d}|^2,|h_{2d}|^2,\cdots,|h_{Md}|^2]$.

As the relays can receive only during the first phase, so the capacity of the source to relay channel is:
\[C_i=\frac{1}{2}\log\left(1+\frac{|h_{si}|^2\alpha  P_1}{\sigma^2+|h_{si}|^2(1-\alpha )P_1} \right)\]
The total power consumed during the second phase can be calculated in following manner:
\begin{equation*}
\alpha  P_1|w_0|^2+(1-\alpha )P_1\left| \sum\limits_{i=1}^{M}\frac{h_{si}h_{id}}{h_{sd}}w_i\right| ^2 +\sum\limits_{i=1}^{M}(|h_{s,i}|^2P_1+\sigma^2)|w_i|^2=\mathbf{w^\d Dw}
\end{equation*}
where $\mathbf{D}=\begin{bmatrix}
\alpha  P_1 & \mathbf{0^t}\\
\mathbf{0} & (1-\alpha )P_1\mathbf{gg^\d}+\mathbf{T}
\end{bmatrix} $, $\mathbf{g}$ is $M$ length vector with $i^{th}$ element $g_i=\frac{h_{si}h_{id}}{h_{sd}}$ and $\mathbf{T}$ is a diagonal matrix with diagonal elements $[|h_{s,1}|^2P_1+\sigma^2,\cdots,|h_{s,M}|^2P_1+\sigma^2]$.
Therefore, the total power constraint can be expressed as
\begin{equation}\label{eq:totcons}
 \mathbf{w^\d Dw} \le P_{tot}
\end{equation}
For several scenarios individual constraint is more relevant than total constraint. Those constraints can be written as:
\begin{subequations}\label{eq:indvcon}
\begin{align}
& \mathbf{w^\d D_sw} \le P_s \\
& |w_i|^2 \le \frac{P_i}{|h_{si}|^2P_1+\sigma^2},\,\forall i\in \{1,2,\cdots,M\}
\end{align}
where $\mathbf{D_s}=\begin{bmatrix}
\alpha  P_1 & \mathbf{0^t}\\
\mathbf{0} & (1-\alpha )P_1\mathbf{gg^\d}
\end{bmatrix}$.
\end{subequations}
Now for a given total or individual power budget maximum achievable secrecy rate can be formulated as following optimization problem:
\begin{subequations}\label{opt1}
\begin{align}
\max\;\min_i&\, R_{si}(\alpha ,\mathbf{w})=C_d-C_i \label{eq:opt1obj}\\
\text{Subject to: } & \eqref{eq:totcons} \text{ or } \eqref{eq:indvcon} \nonumber\\
& \alpha  \in [0,1] \label{eq:opt2alp}
\end{align}
\end{subequations}

We denote the SNR at $i^{th}$ relay node as
 $\Gamma_i=\frac{|h_{si}|^2\alpha  P_1}{\sigma^2+|h_{si}|^2(1-\alpha )P_1}$. One can easily see that $\Gamma_i$ is an increasing function in $|h_{si}|^2, \forall \alpha  \in[0,1]$. Therefore, if we order the source to relay channel gains according to their absolute values, then the one with the highest value will have maximum SNR. Formally,
\[\Gamma_e \ge \Gamma_i, 
\text{ where } e=\arg\max_i\,[|h_{s1}|^2,\ldots,|h_{sM}|^2]\] 
Therefore, the objective of the problem \eqref{opt1} can be rephrased as maximize $C_d-C_e$ for the same constraint set. 
We rewrite $C_d-C_e$ in following manner:
\begin{align*}
 C_d-C_e\hs{-0.1cm} =\frac{1}{2}\left[ \log\left(\frac{\rho_d}{\rho_d-\alpha }+f(\mb{w})\alpha \right)-\log\left(\frac{\rho_e}{\rho_e-\alpha }\right) \right] 
 \end{align*}
 where $ \rho_j\hs{-0.1cm} =\left( \frac{\sigma^2}{|h_{sj}|^2P_1}+1\right), \, j\hs{-0.1cm} \in \{d,e\}$, $f(\mb{w})=\frac{|\mathbf{h}^\d\mathbf{w}|^2P_1}{\sigma^2(1+\mathbf{w^\d D_hw})}$

 By differentiating the objective function with respect to $\alpha $, we have the following observation.
\begin{obv*}
$C_d-C_e$ is an increasing function of $\alpha $ in $[0,1]$ if 
\[f(\mb{w})\ge \frac{(\rho_d-\rho_e)\rho_d}{(\rho_d-1)^2(\rho_e-2)}\]
assuming $\rho_e \ne 2$.
\end{obv*}
But if we use $\alpha =1$, then in absence of artificial noise, relay nodes can properly decode the message. So, we consider an SNR based approach, where an SNR upto $\gamma$ results in significant amount of BER in eavesdropping relay nodes and therefore, for all practical purposes we can assume that the message is undecodable at relay nodes.
The Optimization problem \eqref{opt1} can be rewritten as:
\begin{subequations}\label{opt2}
\begin{align}
\max\, & C_d\\
\text{subject to: }\,&\Gamma_e \le \gamma\\
& \eqref{eq:totcons}\text{ or } \eqref{eq:indvcon},\,\eqref{eq:opt2alp} \nonumber
\end{align}
\end{subequations}
\section{Analysis \& Solution}
From the optimization problem \eqref{opt2} one can easily see that to maximize $C_d$, $\Gamma_e$ should be equal to $\gamma$. Therefore, we can calculate the corresponding $\alpha $ and also using the constraint \eqref{eq:opt2alp} we can bound the $\gamma$ value.
 \[\alpha (\gamma)=\frac{1+\frac{\sigma^2}{|h_{se}|^2P_1}}{(1+\frac{1}{\gamma})},\,\text{ where } \gamma \le \frac{|h_{se}|^2P_1}{\sigma^2}\]
\subsubsection*{Total Power Constraint}
 Once $\alpha $ is fixed the first term in the destination SNR is constant, therefore, optimization problem \eqref{opt2} can be reduced to:
\begin{subequations}\label{optAFtot}
\begin{align}
\max\quad & \frac{|\mathbf{h}^\d\mathbf{w}|^2}{(1+\mathbf{w^\d D_hw})}\\
\text{subject to: } & \mathbf{w^\d Dw} \le P_{tot}
\end{align}
\end{subequations}
One can check that the constraint will be satisfied with equality at optima and the objective function can be written as:
$\frac{\mathbf{w}^\d(\mathbf{hh^\d})\mathbf{w}}{\mathbf{w}^\d\mathbf{\widetilde{D}}\mathbf{w}}$, where $\mathbf{\widetilde{D}}=\frac{1}{P_{tot}}\mathbf{D}+\mathbf{D_h}$. This is indeed the well known generalized \textit{Rayleigh Quotient problem} \cite[p.~176]{Horn1985} for which the maximum value of the objective function corresponds to the maximum eigenvalue of matrix $\mathbf{\widetilde{D}}^{-1}\mathbf{hh^\d}$. 
The solution can be expressed as $\mb{w}^*=\mu\mb{v}$ where $\mb{v}=\mathbf{\widetilde{D}}^{-1}\mb{h}$ and $\mu=\sqrt{\frac{P_{tot}}{\mb{v'Dv}}}$.
\subsubsection*{Individual Power Constraints}: The individual power constraint problem for AF relaying can be written as:
\begin{subequations}\label{optAFindv}
\begin{align}
\max\quad & \frac{|\mathbf{h}^\d\mathbf{w}|^2}{(1+\mathbf{w^\d D_hw})}\\
\text{subject to: } & \eqref{eq:indvcon} \nonumber
\end{align}
\end{subequations}
\par  It can be seen that the angles of the scaling factors should be chosen in following manner to obtain the optimal value.
  \[\arg(w_0)=-\arg(h_{sd}),\,\arg(w_i)=-(\arg(h_{si})+\arg(h_{id})),\forall i\] 
    Therefore, if we denote $\mathbf{c}=[|h_{sd}|,|h_{s1}|,\cdots,|h_{sM}|]^{\mathbf{t}}$ and $\mathbf{u}=[|w_0|,|w_1h_{1d}|,\cdots,|w_Mh_{Md}|]^{\mathbf{t}}$, then the optimization problem \eqref{optAFindv} can be essentially written as: 
    \begin{subequations}\label{opt5}
    \begin{align}
    \max\quad & \frac{(\mathbf{c^tu})^2}{(1+\mathbf{u^tI_0u})}\label{eq:opt5obj}\\
    \text{subject to: }& \mathbf{u^t\t{D}_su} \le P_s \label{eq:opt5srccon}\\
    &0 \le u_i \le u_{max,i},\,\forall i \in \mathcal{M} \label{eq:opt5indvcon}
    \end{align}
    \end{subequations}
  where $\mb{I_0}=\begin{bmatrix}
  0 & \mb{0^t}\\
  \mb{0} & \mb{I}
  \end{bmatrix}$,
  $\mb{\t{D}_s}=\begin{bmatrix}
  \alpha  P_1 & \mathbf{0^t}\\
  \mathbf{0} & \frac{(1-\alpha )P_1}{|h_{sd}|^2}\mathbf{h_sh_s^t}
  \end{bmatrix}$,
  $u_{max,i}=\sqrt{\frac{P_i}{|h_{si}|^2P_1+\sigma^2}}$, $\mb{I}$ is a $M\times M$ identity matrix, and $\mb{h_s}=[|h_{s1}|,|h_{s2}|,\cdots,|h_{sM}|]\trnp$.
  
\begin{prpn}\label{prpn1}
The optimal solution for problem \eqref{opt5} with only source constraint i.e. \eqref{eq:opt5srccon} is given by:
$u_1^*=\sqrt{\eta_1-\eta_2||\mb{c}_{\<2\>}||^2r^{*2}}$, $\mb{u}_{\<2\>}^*=\frac{\mb{c}_{\<2\>}}{||\mb{c}_{\<2\>}||}r^*$ where $r^*=\sqrt{\frac{||\mb{c}_{\<2\>}||^2\eta_1}{||\mb{c}_{\<2\>}||^4\eta_2+(\eta_1+||\mb{c}_{\<2\>}||^2\eta_2)^2c_1^2}}$
\end{prpn}
\begin{proof}
We denote the vector $[u_2,u_3,\cdots,u_{M+1}]$ as $\mb{u}_{\< 2\>}$, 
then the objective function can be rewritten as:
\[\frac{(c_1u_1+\mb{c}_{\<2\>}\trnp\mb{u}_{\<2\>})^2}{1+||\mb{u}_{\<2\>}||^2}\]
\par Now it is easy to see that, if $||\mb{u}_{\<2\>}||=r$, then maximum value of  $\mb{c}_{\<2\>}\trnp\mb{u}_{\<2\>}=\frac{\mb{c}_{\<2\>}}{||\mb{c}_{\<2\>}||}r$. As the objective function is an increasing function of $u_1$, so in absence of individual constraints \textit{i.e.} \eqref{eq:opt5indvcon}, the source constraint \eqref{eq:opt5srccon} will be satisfied with equality. From that we can calculate the value of $u_1$ in terms of $r$.
\[u_1=\sqrt{\eta_1-\eta_2||\mb{c}_{\<2\>}||^2r^2},\, \text{ where } \eta_1=\frac{P_s}{\alpha  P_1},\, \eta_2=\frac{(1-\alpha )}{\alpha |c_1|^2} \]
We can write the objective function \eqref{eq:opt5obj} in terms of $r$ in following manner: $\frac{(c_1\sqrt{\eta_1-\eta_2r^2}+||\mb{c}_{\<2\>}||r)^2}{1+r^2}$
By differentiating with respect to $r$ and equating it to 0, we get the above result.
\end{proof}
If the solution obtained in the above manner satisfy all the individual constraints, then we have indeed solved the problem \eqref{opt5}.
But, if for any $u_i$ individual constraint is violated, then we solve the problem using the following iterative approach. 
\begin{enumerate}
\item Initialize parameters $t_1=0$ and $t_2=1$.
\item Sort the relay nodes in descending order based on $u_i^*/u_{max,i},\,\forall i\in \{2,\cdots,M\}$. Sort the vector $\mb{c}_{\<2\>}$ and $\mb{u}_{max}$ accordingly. In sorted list let us denote the first variable as $u_{(2)}$. Update the parameter $t_1=t_1+c_{(2)}u_{(2)}$ and $t_2=t_2+u_{(2)}^2$. Now we have to solve the following optimization problem:
\begin{subequations}\label{opt6}
\begin{align}
\max\quad & \frac{(t_1+c_1u_1+\mb{\t{c}\trnp\t{u})^2}}{t_2+||\mb{\t{u}}||^2}\\
\text{ s.t. }& \alpha  P_1u_1^2+\frac{(1-\alpha )P_1}{|h_{sd}|^2}(t_1+\mb{\t{c}\trnp\t{u}})^2=P_s\label{eq:opt6cons}
\end{align}
\end{subequations}
where $\mb{\t{u}}$ and $\mb{\t{c}}$ are the vectors obtained after removing $u_{(2)}$ and $c_{(2)}$, respectively.
\item If the solution of problem \eqref{opt6} satisfies all the individual constraints, then optimal solution is obtained, otherwise repeat step 2 \& 3 until the solution obtained satisfy their individual constraints. 
\end{enumerate}
\textit{Remark:} The rationale behind picking the variable $u_{(2)}$ corresponding to $\arg\max_i\, u_i^*/u_{max,i}$ in step 2 is: for increasing $P_s$, $u_{(2)}$ is the first variable to violate its individual constraint and \eqref{eq:opt5obj} increases with $u_{(2)}$ till it reaches the optimum corresponding to that $P_s$ value.
\begin{prpn}
The optimal solution of the problem \eqref{opt6} is given by $u_1^*=\sqrt{\eta_1-\eta_2(t_1+||\mb{\t{u}}||r^*)^2}$, $\mb{\t{u}}^*=\frac{\mb{\t{c}}}{||\mb{\t{c}}||}r^*$ where $r^*$ satisfies the following quadratic polynomial:
\begin{equation}\label{eq:quad}
q_0r^4+q_1r^3+q_2r^2+q_3r+q_4=0
\end{equation}
where $q_0=\eta_2\eta_3t_1^2\tau^2$, $q_1=-2\eta_2t_1\tau(\eta_1c_1^2+\eta_3(t_2\tau^2-t_1^2))$, $q_2=\eta_3(\eta_1t_1^2-\eta_2((t_1^2-t_2\tau^2)^2-2t_1^2t_2\tau^2))-\eta_1c_1^2(\eta_1-\eta_2t_1^2+2\eta_2t_2\tau^2)$, $q_3=2t_1t_2\tau\eta_3(\eta_1-\eta_2t_1^2+\eta_2t_2\tau^2)$, $q_4=t_2^2\tau^2(\eta_1-\eta_2\eta_3t_1^2)$  $\eta_3=(1+\eta_2c_1^2)$, $\tau=||\mb{\t{c}}||$
\end{prpn}
\begin{proof}
Following the argument of Proposition \eqref{prpn1} we can write $\mb{\t{c}\trnp\t{u}}=||\mb{\t{c}}||r$, where $||\mb{\t{u}}||=r$. Now due to equality constraint \eqref{eq:opt6cons}, we can express $u_1$ in terms of $r$ as: $u_1=\sqrt{\eta_1-\eta_2(t_1+||\mb{\t{c}}||r)^2}$ Therefore, optimization problem \eqref{opt6} can be rewritten as:
\begin{equation}\label{opt7}
\max_r\, \frac{(t_1+||\mb{\t{c}}||r+c_1\sqrt{\eta_1-\eta_2(t_1+||\mb{\t{c}}||r)^2})^2}{t_2+r^2}
\end{equation}
By differentiating it with respect to $r$ and equating it to 0, we get the above polynomial equation in terms of $r$.
\end{proof}
\begin{cor}
If $r^*$ is the real positive root of equation \eqref{eq:quad} which maximizes the objective function of problem \eqref{opt7}, then the solution of problem \eqref{opt6} is obtained using that $r^*$.
\end{cor}
\begin{proof}
As $u_1$ and $\mb{\t{u}}$ calculated using $r^*$ satisfy the constraint \eqref{eq:opt6cons} and maximizes \eqref{opt7}, hence it is indeed the optimal solution of problem \eqref{opt6}.
\end{proof}
\section{Results}
In Figure \ref{fig:CdvsP1} we plot the capacity ($C_d$) of source to destination channel with respect to transmit power at first stage ($P_1$) for multiple $\alpha $ values. Source to relay and relay to destination channel gains were generated from complex Gaussian distribution with mean 0 and variance 1. The distribution used for source to destination channel is $\mathcal{CN}(0,0.25)$. While plotting the results we averaged \ignore{took the average of} $C_d$ over 100 such network instances. For individual power constraints we considered $P_s=5$ and $P_i=0.1, \forall i$, whereas, for total power constraint we used $P_{tot}=P_s+M*P_i$. As the SNR value at the destination is an increasing function of $P_1$ and $\alpha $, so the channel capacity increases as we increase both the parameters. 
\begin{figure}[!htbp]
\centering
\includegraphics[scale=0.6]{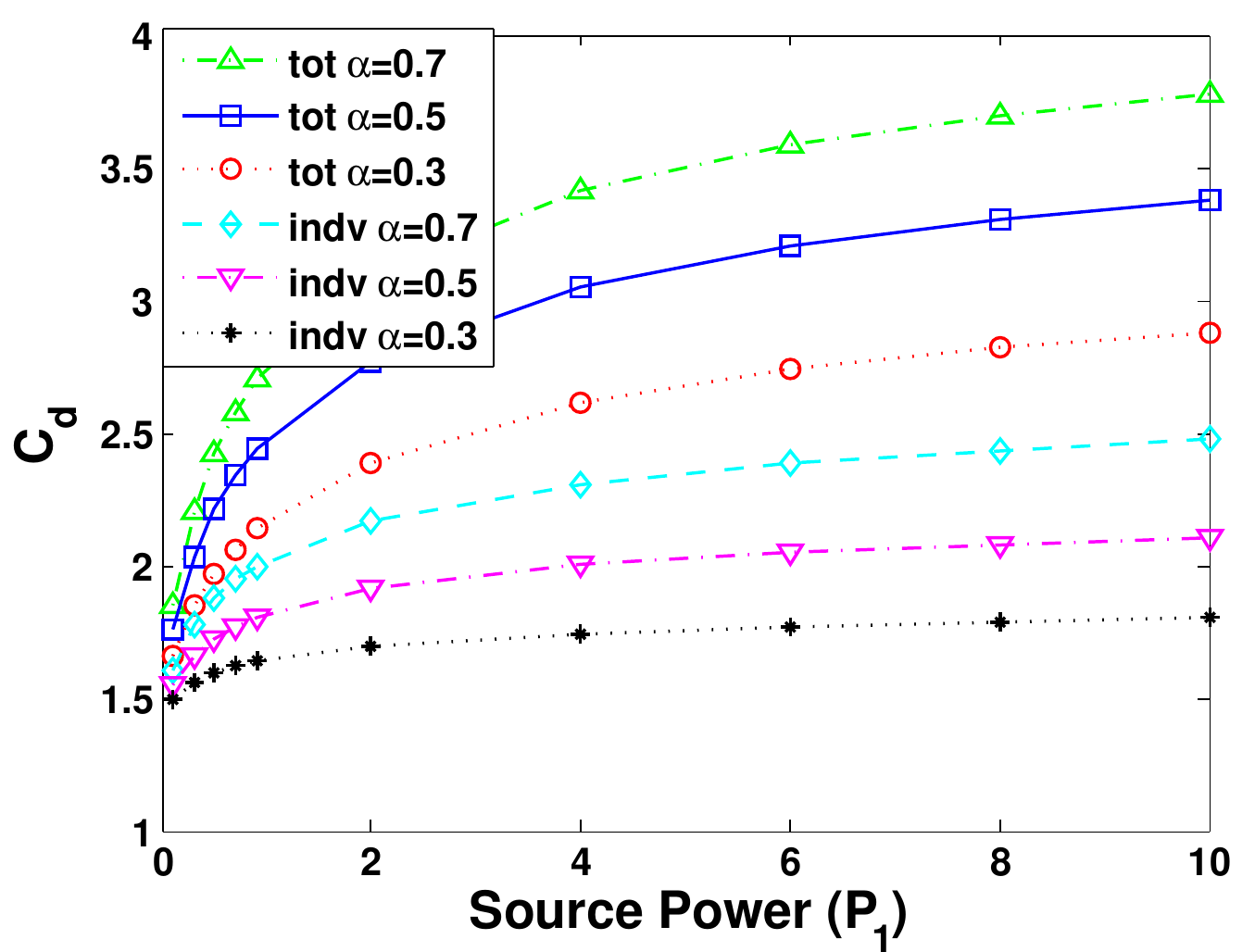}
\caption{Plot of $C_d$ with respect to $P_1$ for several $\alpha $ values in case of both total and individual constraints.}
\label{fig:CdvsP1}
\end{figure}
In Figure \ref{fig:CdvsM} we plot the capacity for both total and individual constraint scenario with respect to number of the relay nodes. As the number of relay nodes increases the second term in the destination SNR also increases, which results in increase of $C_d$.
\begin{figure}[!htbp]
\centering
\includegraphics[scale=0.6]{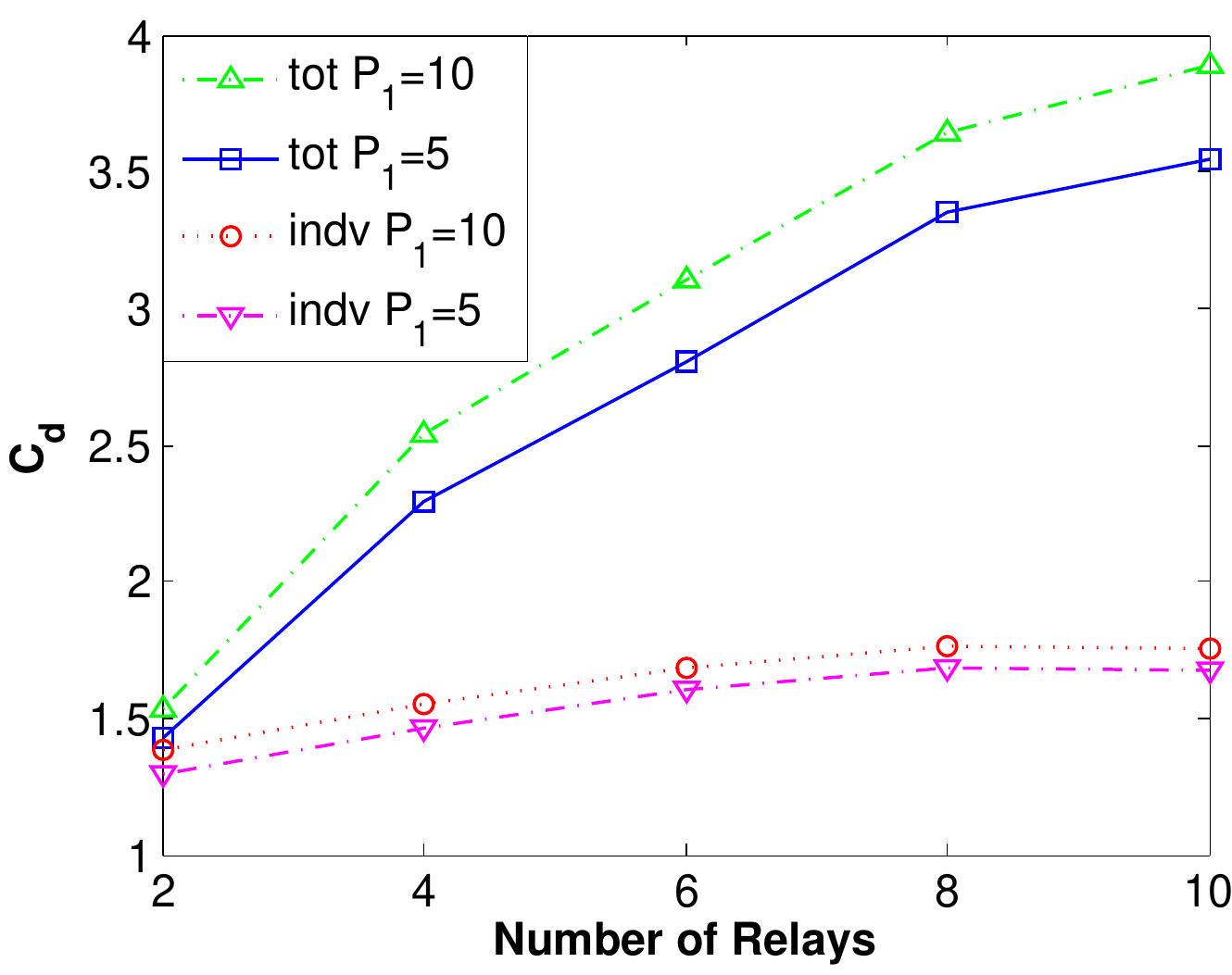}
\caption{Plot of $C_d$ with respect to number of relays for several $P_1$ values in case of both total and individual constraints.}
\label{fig:CdvsM}
\end{figure}
\section{Conclusion} In the current article we have presented a beamforming solution for secure communication using untrusted relay nodes. After justifying SNR based approach, we formulated and solved two optimization problems -- one for total power constraint and other for individual power constraints. Our current model assumes a perfect CSI for evaluation of optimal scaling vector which may not be available in several practical scenarios. Therefore, in future we would like to study the performance of our model for delayed CSI.


\begin{thebibliography}{10}
\providecommand{\url}[1]{#1}
\csname url@samestyle\endcsname
\providecommand{\newblock}{\relax}
\providecommand{\bibinfo}[2]{#2}
\providecommand{\BIBentrySTDinterwordspacing}{\spaceskip=0pt\relax}
\providecommand{\BIBentryALTinterwordstretchfactor}{4}
\providecommand{\BIBentryALTinterwordspacing}{\spaceskip=\fontdimen2\font plus
\BIBentryALTinterwordstretchfactor\fontdimen3\font minus
  \fontdimen4\font\relax}
\providecommand{\BIBforeignlanguage}[2]{{%
\expandafter\ifx\csname l@#1\endcsname\relax
\typeout{** WARNING: IEEEtran.bst: No hyphenation pattern has been}%
\typeout{** loaded for the language `#1'. Using the pattern for}%
\typeout{** the default language instead.}%
\else
\language=\csname l@#1\endcsname
\fi
#2}}
\providecommand{\BIBdecl}{\relax}
\BIBdecl

\bibitem{wyner}
A.~Wyner, ``{The Wire-tap Channel},'' \emph{Bell Systems Technical Journal},
  vol.~54, no.~8, pp. 1355--1387, Jan 1975.

\bibitem{Leung}
S.~Leung-Yan-Cheong and M.~Hellman, ``{The Gaussian Wire-tap Channel},''
  \emph{IEEE Transactions on Information Theory}, vol.~24, no.~4, pp. 451--456,
  Jul 1978.

\bibitem{dong}
{Lun Dong and Zhu Han and Petropulu, A.P. and Poor, H.V.},
  ``Amplify-and-forward based cooperation for secure wireless communications,''
  in \emph{Acoustics, Speech and Signal Processing, 2009. ICASSP 2009. IEEE
  International Conference on}, 2009, pp. 2613--2616.

\bibitem{zhang10}
J.~Zhang and M.~C. Gursoy, ``Collaborative relay beamforming for secrecy,'' in
  \emph{Communications (ICC), 2010 IEEE International Conference on}.\hskip 1em
  plus 0.5em minus 0.4em\relax IEEE, 2010, pp. 1--5.

\bibitem{xiang}
X.~He and A.~Yener, ``{Two-Hop Secure Communication Using an Untrusted Relay: A
  Case for Cooperative Jamming},'' in \emph{IEEE Global Telecommunications
  Conference (GLOBECOM 08)}, Dec 2008, pp. 1--5.

\bibitem{zhang}
R.~Zhang, L.~Song, Z.~Han, B.~Jiao, and M.~Debbah, ``{Physical Layer Security
  for Two Way Relay Communications with Friendly Jammers},'' in \emph{IEEE
  Global Telecommunications Conference (GLOBECOM 2010)}, Dec 2010, pp. 1--6.

\bibitem{jeong12}
C.~Jeong, I.-M. Kim, and D.~I. Kim, ``Joint secure beamforming design at the
  source and the relay for an amplify-and-forward mimo untrusted relay
  system,'' \emph{IEEE Transactions on Signal Processing}, vol.~60, no.~1, pp.
  310--325, Jan 2012.

\bibitem{he2010cooperation}
X.~He and A.~Yener, ``{Cooperation with an Untrusted Relay: A Secrecy
  Perspective},'' \emph{IEEE Transactions on Information Theory}, vol.~56,
  no.~8, pp. 3807--3827, 2010.

\bibitem{goel}
R.~Negi and S.~Goel, ``{Secret Communication using Artificial Noise},'' in
  \emph{Proceeding IEEE 62nd Vehicular Technology Conference (VTC-2005-Fall)},
  vol.~3, Sep 2005, pp. 1906--1910.

\bibitem{yang2013}
Y.~Yang, Q.~Li, W.-K. Ma, J.~Ge, and M.~Lin, ``{Optimal Joint Cooperative
  Beamforming and Artificial Noise Design for Secrecy Rate Maximization in AF
  Relay Networks},'' in \emph{IEEE 14th Workshop on Signal Processing Advances
  in Wireless Communications (SPAWC), 2013}, June 2013, pp. 360--364.

\bibitem{Horn1985}
R.~A. Horn and C.~R. Johnson, Eds., \emph{Matrix Analysis}.\hskip 1em plus
  0.5em minus 0.4em\relax New York, NY, USA: Cambridge University Press, 1986.

\end{thebibliography}
\end{document}